\documentclass[11pt]{amsart}
\usepackage{amssymb,mathrsfs,graphicx,enumerate}
\usepackage{amsmath,amsfonts,amssymb,amscd,amsthm,bbm}
\usepackage{kotex}
\allowdisplaybreaks
\usepackage[retainorgcmds]{IEEEtrantools}
\usepackage{colortbl}
\usepackage{tikz}
\usepackage{tikz-3dplot}
\usepackage{pgfplots}

\title[CS model on quotient spaces]{A velocity alignment model on quotient spaces of the Euclidean space}

\author[Park]{Hansol Park}
\address[Hansol Park]{\newline Department of Mathematical Sciences\newline Seoul National University, Seoul 08826, Republic of Korea}
\email{hansol960612@snu.ac.kr}

\newtheorem{theorem}{Theorem}[section]
\newtheorem{lemma}{Lemma}[section]
\newtheorem{corollary}{Corollary}[section]
\newtheorem{proposition}{Proposition}[section]
\newtheorem{remark}{Remark}[section]
\newtheorem{example}{Example}[section]

\newcommand{\bbr}{\mathbb R}

\newcommand{\bbs}{\mathbb S}
\newcommand{\bbz}{\mathbb Z}

\newcommand{\bbt} {\mathbb T}
\newcommand{\bbm} {\mathbb M}
\newcommand{\bbk} {\mathbb K}

\allowdisplaybreaks

\makeatletter
\@namedef{subjclassname@2020}{\textup{2020} Mathematics Subject Classification}
\makeatother

\begin{document}

\date{\today}


\thanks{\textbf{Acknowledgment.} The work of H. Park is supported by NRF-2020R1A2C3A01003881(National Research Foundation of Korea).}

\begin{abstract}
The Cucker-Smale(CS) model is a velocity alignment model, and this model also has been generalized on general manifolds. We modify the CS model on manifolds to get rid of a-priori condition on particles' positions and conditions on communication functions. Since the shortest geodesic is used to define an interaction between two particles, if there exist two or more than two shortest geodesics, then the system is not well-defined. In this paper, instead of using the shortest geodesic to define an interaction between two particles, we use all geodesics to define an interaction. From this assumption, we can relax the a-priori condition and conditions on communication functions. We also explain the relationship between the suggested model and previous models. Finally, we provide some emergent behaviors on some specific manifolds(e.g. flat torus, flat M\"{o}bius strip, and flat Klein bottle). From these results, we can discuss the effect of the topology of the domain.
\end{abstract}

\maketitle \centerline{\date}

\section{Introduction}\label{sec:1}
Analyzing the emergent behavior of dynamical systems is a one of the important part of the applied mathematics. It can be used to analyze the following phenomenon: flashing of fireflies \cite{ballerini2008interaction}, schooling of fish \cite{Aoki1982, degond2008large, Topaz:Bertozzi}. Also, it can be applied to unmanned aerial vehicles \cite{albi2019vehicular}, cooperative robot systems \cite{gamba2016global, markdahl2020robust, wang2018decentralized},  variational method \cite{Figalli_etal2011, ChFeTo2015, FeZh2019}, shape matching problem \cite{ha2020dynamical}, rotation averaging problem \cite{kapic2021new}, and minimization problem \cite{markdahl2021counterexamples}. In this paper, we study a velocity alignment model intensively. One of the famous velocity alignment model is the Cucker-Smale(CS) model studied in \cite{Cucker:Smale}, and this model is given as follows:
\begin{align}\label{A-1}
\begin{cases}
\dot{x}_i=v_i,\quad t>0,\\
\displaystyle\dot{v}_i=\frac{\kappa}{N}\sum_{j=1}^N\psi(x_i, x_j)(v_j-v_i),\quad i=1, 2, \cdots, N,\\
x_i(0)=x_i^0\in\bbr^d,\quad v_i(0)=v_i^0\in\bbr^d,
\end{cases}
\end{align}
where $\kappa>0$ is the coupling strength and $\psi$ is the communication function. Under proper conditions, the velocity alignment of this system, i.e.
\[
\lim_{t\to\infty}\|v_i-v_j\|=0\quad\forall~1\leq i, j\leq N,
\]
occurs. System \eqref{A-1} is cefined on the Euclidean space. A canonical generalization of this work is extension of the model on general manifolds. This extension is studied in \cite{ha2020emergent}, and the CS model on general manifolds is given as follows:
\begin{align}\label{A-2}
\begin{cases}
\displaystyle\frac{d}{dt}x_i=v_i\vspace{0.2cm}\\
\displaystyle\frac{D}{dt}v_i=\frac{\kappa}{N}\sum_{k=1}^N\psi(x_i, x_k)(P_{ik}v_k-v_i),\\
(x_i(0), v_i(0))=(x_i^0, v_i^0)\in TM\quad\forall i\in\{1, 2, \cdots, N\},
\end{cases}
\end{align}
where $P_{ik}$ is a parallel transport of a tangent vector on $T_{x_k}M$ to a tangent vector on $T_{x_i}M$ and $\frac{D}{dt}$ is a covariant derivative on $M$. This model was studied on the sphere \cite{ha2020emergent, ahn2021emergent}, the hyperbolic space $\mathbb{H}^2$ \cite{ahn2020emergent}, the special orthogonal group $SO(3)$ \cite{fetecau2021emergent}. A disadvantage of this model is that the interaction between $i^{th}$ and $j^{th}$ particles can not be well-defined if there exist two or more than two shortest geodesics which connect $x_i$ and $x_j$. To prevent this situation, the authors of \cite{ha2020emergent} assumed the following a-priori condition:
\begin{center}
``A shortest geodesic between two points $x_i(t)$ and $x_j(t)$ is unique for all $i\neq j$, $t\geq0$."
\end{center}
Actually, necessity of this a-priori condition comes from the structure of system \eqref{A-2}. The shortest geodesic is special in this system, since \textit{the shortest geodesic} between $x_i$ and $x_j$ is used to define an interaction between $i^{th}$ and $j^{th}$ particles. In this paper, we are interested in the following question:\vspace{0.2cm}

\begin{center}
$\bullet$(Q): How can we remove the speciality of the shortest geodesic between $x_i$ and $x_j$?\vspace{0.2cm}
\end{center}

To remove the speciality of the shortest geodesic, we used all geodesics which connect $x_i$ and $x_j$ to define the interaction between $i^{th}$ and $j^{th}$ particles. We suggest a modified system on a manifold $M$ as follows:
\begin{align*}
\begin{cases}
\displaystyle\frac{d}{dt}x_i=v_i,\vspace{0.2cm}\\
\displaystyle\frac{D}{dt}v_i=\frac{\kappa}{N}\sum_{k=1}^N\sum_{\gamma\in\Gamma_{x_k}^{x_i}}\varphi(|\gamma|)(P_{ik}^\gamma v_k-v_i),\quad t>0,\\
x_i(0)=x_i^0\in M,\quad v_i(0)=v_i^0\in T_{x_i^0}M,\quad \forall i\in \mathcal{N},
\end{cases}
\end{align*}
where $\Gamma_{x_k}^{x_i}$ is a set of all geodesics which starts from $x_k$ and finish at $x_i$, and $P_{ik}^\gamma$ is a parallel transport from $x_k$ to $x_i$ along a geodesic $\gamma$. Here, $\varphi$ is a communication function depends on the length of a geodesic $\gamma$. Throughout this paper, we denote the universal covering space of $M$ with the covering metric by $\tilde{M}$ and $p:\tilde{M}\to M$ is the corresponding covering map. We can lift this system onto the universal covering space $\tilde{M}$ as follows:
\begin{align*}
\begin{cases}
\displaystyle\frac{d}{dt}\tilde{x}_i=\tilde{v}_i,\vspace{0.2cm}\\
\displaystyle\frac{D}{dt}\tilde{v}_i=\frac{\kappa}{N}\sum_{k=1}^N\sum_{\tilde{y}_k\in p^{-1}(p(\tilde{x}_k))} \varphi(\mathrm{dist}(\tilde{x}_i, \tilde{y}_k))(P_{\tilde{x}_i\tilde{y}_k} \tilde{u}^{\tilde{y}_k}_k-\tilde{v}_i),\quad t>0,\\
\tilde{x}_i(0)=\tilde{x}_i^0\in\tilde{M}^d,\quad \tilde{v}_i(0)=\tilde{v}_i^0\in T_{\tilde{x}_i^0}\tilde{M},\quad\forall~i\in\mathcal{N},
\end{cases}
\end{align*} 
where $p^{-1}_{\tilde{x}}$ is a inverse map of $p$ defined on neighborhood of $\tilde{x}$ and $\tilde{u}_k^{\tilde{y}_k}:=Dp^{-1}_{\tilde{y}_k}(v_k)=Dp^{-1}_{\tilde{y}_k}\circ Dp_{\tilde{x}_k}(\tilde{v}_k)$ and $P_{\tilde{x}\tilde{y}}$ is a parallel transport from a tangent vector at $\tilde{y}$ to a tangent vector at $\tilde{x}$. We assume the following two conditions to $M$:\\

\noindent($\mathcal{M}1$): A set $p^{-1}(x)\subset \tilde{M}$ is at most countable set for any $x\in M$.\\

\noindent($\mathcal{M}2$): For any points $\tilde{x}, \tilde{y}\in\tilde{M}$, there exists a unique geodesic which connects two points $\tilde{x}$ and $\tilde{y}$.\\

Here, ($\mathcal{M}1$) is assumed for the well-definedness of sum $\sum_{\tilde{y}_k\in p^{-1}(p(\tilde{x}_k))}$, and ($\mathcal{M}2$) is assumed for the well-definedness of the parallel transport $P_{\tilde{x}_i\tilde{y}_k}$. We provide more details about these assumptions in Section \ref{sec:4.1}.

Now, we define the following energy functional
\[
\mathcal{E}[\mathcal{V}]:=\frac{1}{2}\sum_{i=1}^N\|v_i\|^2
\]
where $\mathcal{V}:=\{v_i\}_{i=1}^N$. Then $\mathcal{E}$ decreases along the time evolution(i.e. $\frac{d}{dt}\mathcal{E}\leq0$). We combine this fact and Barbalat's lemma(Lemma \ref{L2.1}) to obtain the long-time behavior of the system as follows(See Theorem \ref{T4.1}):
\[
\lim_{t\to\infty}\sum_{\gamma\in \Gamma_{x_k}^{x_i}}\psi(|\gamma|)\|P_{ik}^\gamma v_k-v_i\|^2=0\quad\forall~1\leq i, k\leq N.
\]
If we lift this result to the universal covering space $\tilde{M}$, we have the follows(See Corollary \ref{C4.1}):
\begin{align}\label{A-5}
\lim_{t\to\infty}\sum_{\tilde{y}_k\in p^{-1}(p(\tilde{x}_k))}\varphi(\mathrm{dist}(\tilde{x}_i, \tilde{y}_k))\|P_{\tilde{x}_i\tilde{y}_k}\tilde{u}_k^{\tilde{y}_k}-\tilde{v}_i\|^2=0\quad\forall~1\leq i, k\leq N.
\end{align}
Since $P_{\tilde{x}_i\tilde{y}_k}\tilde{u}_k^{\tilde{y}_k}$ can not be simplified on a general manifold $\tilde{M}$, we assume that $\tilde{M}$ is the Euclidean space. Also, we simplify condition \eqref{A-5} on some specific manifolds(e.g. flat torus, flat M\"{o}bius strip, and flat Klein bottle) in Section \ref{sec:5}. On the flat torus, we could obtain the ordinary velocity alignment 
\begin{align}\label{A-6}
\lim_{t\to\infty}\|v_i-v_j\|=0\quad\forall~1\leq i, k\leq N,
\end{align}
under suitable conditions(see Section \ref{sec:5.1}). Let the flat M\"{o}bius strip and the flat Klein bottle constructed by the way introduced in Sections \ref{sec:5.2} and \ref{sec:5.3}, respectively. Then, on the flat M\"{o}bius strip and the flat Klein bottle, we could obtain the ordinary velocity alignment \eqref{A-6} and the second component of velocities $v_i$ converges for all $1\leq i\leq N$(see Sections \ref{sec:5.2} and \ref{sec:5.3}).
\\

The rest of this paper is organized as follows. In Section \ref{sec:2}, we introduce the universal covering space and previous results of the Cucker-Smale(CS) model on manifolds. In Section \ref{sec:3}, we modify the CS model on manifold on the flat torus $\bbt^d$.  From this modification, we suggest the modified CS model on manifold on general manifold under some assumption in Section \ref{sec:4}. We also study emergent behaviors of the modified CS model on some specific spaces(flat torus, flat M\"{o}bius strip, flat Klein bottle) in Section \ref{sec:5}. Finally, Section \ref{sec:6} is devoted to a brief summary of the paper.\\

\noindent\textbf{Gallery of Notations.} Now, we present some notations. Since we use the index set frequently, we define
\[
\mathcal{N}:=\{1, 2, \cdots, N\}.
\]
Also, for any vector $x\in\bbr^d$, we denote $\alpha^{th}$ component of $x$ by $(x)_\alpha$. 
\[
\text{i.e.}\quad x=((x)_1, (x)_2, \cdots, (x)_d).
\]

\section{Preliminaries}\label{sec:2}
In this section, we provide some preparatory concepts. Since the goal of this paper is to construct a velocity alignment model using universal covering spaces, we provide a simple review of universal covering spaces and prior studies on velocity alignment models.

\subsection{Universal covering space of spaces of constant curvature}\label{sec:2.1}
In this subsection, we provide a review on universal covering spaces of constant curvature space, since we will consider quotient spaces of the Euclidean space.

\begin{theorem}[Theorem 4.1 of \cite{doCarmo1992}]\label{T2.1}
Let $M$ be a $d$-dimensional complete Riemannian manifold with constant sectional curvature zero. Then the universal covering $\tilde{M}$ of $M$, with the covering metric, is isometric to $\bbr^d$.
\end{theorem}
Here, covering metric means that the covering map $p:\tilde{M}\to M$ is a local isomorphism.
\begin{corollary}
From Theorem \ref{T2.1}, we have the following results:

\noindent(1) The universal covering space of the flat torus $\mathbb{T}^d$ is $\bbr^d$.

\noindent(2) The universal covering spaces of the flat M\"{o}bius strip $\mathbb{M}$ and the flat Klein bottle $\mathbb{K}$ are $\bbr^2$.
\end{corollary}

\subsection{The Cucker-Smale model on manifolds}\label{sec:2.2}
In this subsection, we introduce the Cucker-Smale(CS) model on manifolds and its emergent behaviors. System \eqref{A-1} is a velocity alignment model on the Euclidean space introduced in \cite{Cucker:Smale}. It is natural that generalization of this system onto general manifolds. The authors of \cite{ha2020emergent} suggested the CS model on general manifolds as system \eqref{A-2}. Now, we present the previous results of system \eqref{A-2} on various spaces.

\subsubsection{On the sphere $\bbs^d$}
The CS model \eqref{A-2} on sphere was first studied in \cite{ha2020emergent}. In this paper, the authors assumed the following a-priori:
\begin{align}\label{B-1}
\sup_{0\leq t<\infty} \max_{i, j\in \mathcal{N}} \mathrm{dist}(x_i(t), x_k(t))<\pi.
\end{align}
This a priori assumption allows the well-posedness of the system on the sphere $\bbs^d$. The authors of \cite{ahn2021emergent} improved these result. They assumed the following condition on the communication function $\psi$:
\begin{align}\label{B-2}
\mathrm{dist}(x, y)=\pi\quad\Longrightarrow\quad\psi(x, y)=0.
\end{align}
The well-posedness of system \eqref{A-2} on the sphere can be broken when the parallel transport $P_{ik}$ is not well defined, and this situation occurs when $x_i$ and $x_k$ are the antipodal points of each others for some $i, k\in\mathcal{N}$. However, if we impose condition \eqref{B-2}, we do not have to calculate $P_{ik}v_k-v_i$ since $\psi(x_i, x_k)$ is zero. One of the useful lemma to show a velocity alignment is Barbalat's lemma. We provide the lemma without the proof.
\begin{lemma}[Barbalat's lemma \cite{Barbalat1959}]\label{L2.1}
Suppose that a real-valued function $f:[0, \infty)\to\bbr$ is uniformly continuous and it satisfies
\[
\lim_{t\to\infty}\int_0^t f(s)ds\quad\text{exists}.
\]
Then, $f$ tends to zero as $t\to\infty$. i.e. $\lim_{t\to\infty}f(t)=0$.
\end{lemma}

\begin{theorem}[Emergent behavior on the sphere \cite{ahn2021emergent}]\label{T2.2}
Let $\{(x_i, v_i)\}_{i=1}^N$ be a global smooth solution to \eqref{A-2} on the sphere $\bbs^d$ and assume that $\psi:\bbs^d \times \bbs^d\to\bbr$ is a positive smooth function satisfying \eqref{B-2}. Then, we have the following dichotomy for the asymptotic dynamics of $\{(x_i, v_i)\}_{i=1}^N$:

\noindent(1) either the energy converges to zero:
\[
\lim_{t\to\infty}\mathcal{E}(t)=0,
\]
where $\mathcal{E}(t)=\sum_{i=1}^N\|v_i\|^2$.

\noindent(2) or the energy converges to a nonzero positive value and all positions approach to a common great circle asymptotically: for every $i, j, k\in\mathcal{N}$ and $a, b, c\in\{1, \cdots, d+1\}$ we have
\[
\lim_{t\to\infty}\mathcal{E}(t)>0\quad\text{and}\quad \lim_{t\to\infty}\mathrm{det}\begin{pmatrix}
x_i^a(t)&x_j^a(t)&x_k^a(t)\\
x_i^b(t)&x_j^b(t)&x_k^b(t)\\
x_i^c(t)&x_j^c(t)&x_k^c(t)
\end{pmatrix}=0.
\]
\end{theorem}

\subsubsection{On the hyperbolic space $\mathbb{H}^2$}
The CS model on the hyperbolic space $\mathbb{H}^2$ was studied in \cite{ahn2020emergent}. For any $x, y\in\mathbb{H}^2$, there exists the unique geodesic $\gamma$ which connects $x$ and $y$. This means, the parallel transport $P_{ik}$ between $x_i$ and $x_k$ is always well-defined. So, we do not need a condition which is similar to \eqref{B-2} in this case.

\begin{theorem}[Emergent behavior on the hyperbolic space \cite{ahn2020emergent}]\label{T2.3}
Let $\{(x_i, v_i)\}_{i=1}^N$ be a global smooth solution to \eqref{A-2} on the hyperbolic space $\mathbb{H}^2$ and assume that $\psi:\mathbb{H}^2\times \mathbb{H}^2\to\bbr$ is a strictly positive smooth function. Then, we have the following dichotomy for the asymptotic dynamics of $\{(x_i, v_i)\}_{i=1}^N$:

\noindent(1) either the energy tends to zero:
\[
\lim_{t\to\infty}\mathcal{E}(t)=0.
\]

\noindent(2) or the energy converges to a positive value, and position configuration becomes coplanar asymptotically:
\[
\lim_{t\to\infty}\mathcal{E}(t)=\mathcal{E}^\infty>0\quad\text{and}\quad \lim_{t\to\infty}\mathrm{det}\Big(x_i(t)~\Big|~ x_j(t)~\Big| x_k(t)\Big)=0\quad\forall~i, j, k\in\mathcal{N}
\]
where $x_i(t)\in\mathbb{H}^2\subset\bbr^3$ is considered as a three-dimensional vector for all $i\in\mathcal{N}$.
\end{theorem}

\subsubsection{On the special orthogonal group $SO(3)$}
The CS model on $SO(3)$ was studied in \cite{fetecau2021emergent}. Instead of using $\{(x_i, v_i)\}_{i=1}^N$, we use the pair $\{(R_i, A_i)\}_{i=1}^N$ which satisfies
\[
x_i=R_i,\quad v_i=R_iA_i\quad~\forall~ i\in\mathcal{N}.
\]
Here, $R_i$ is an element of $SO(3)$ and $A_i$ is a skew symmetric matrix of size $3\times 3$. Then, we can express system \eqref{A-2} on $SO(3)$ as follows:
\begin{align}\label{B-5}
\begin{cases}
\displaystyle\frac{d}{dt}R_i=R_iA_i,\\
\displaystyle\frac{d}{dt}\mathbf{a}_i=\frac{\kappa}{N}\sum_{k=1}^N\varphi(R_i, R_k)\left[
\left(1-\cos\frac{\theta_{ki}}{2}\right)(\mathbf{n}_{ki}\cdot\mathbf{a}_k)\mathbf{n}_{ki}+\sin\frac{\theta_{ki}}{2}\mathbf{a}_k\times\mathbf{n}_{ki}+\cos\frac{\theta_{ki}}{2}\mathbf{a}_k-\mathbf{a}_i,
\right]
\end{cases}
\end{align}
where 
\[
\hat{\mathbf{n}}_{ki}=\frac{\hat{\mathbf{u}}_{ki}}{\theta_{ki}},\quad \mathbf{a}_i=\check{A}_i
\]
and
\[
\theta_{ki}=\arccos\left(\frac{\mathrm{tr}(R_k^\top R_i)-1}{2}\right),\quad \hat{\mathbf{u}}_{ki}=\frac{\theta_{ki}}{2\sin\theta_{ki}}(R_k^\top R_i-R_i^\top R_k).
\]
Here, $\hat{\cdot}$ and $\check{\cdot}$ operators are defined as follows:
\[
\hat{x}:=A,\quad \check{A}:=x,
\]
where
\[
x=(x_1, x_2, x_3)\in\bbr^3,\quad A=\begin{pmatrix}
0&-x_3&x_2\\
x_3&0&-x_1\\
-x_2&x_1&0
\end{pmatrix}\in\mathfrak{so}(3).
\]

\begin{theorem}[Emergent behavior on the special orthogonal group \cite{fetecau2021emergent}]\label{T2.4}
Let $\{(R_i, A_i)\}_{i=1}^N$ be a solution to system \eqref{B-5}. Then, we have the following dichotomy for the asymptotic dynamics of $\{(R_i, A_i)\}_{i=1}^N$:

\noindent(1) either the kinetic energy tends to zero:
\[
\lim_{t\to\infty}\mathcal{E}(t)=0.
\]

\noindent(2) or the energy converges to positive value $\mathcal{E}^\infty$, and 
\[
\lim_{t\to\infty}(A_i(t)-A_k(t))=0,\quad \forall~i, k\in\mathcal{N}.
\]
Also, particles approach and rotate with constant speed along a common geodesic.
\end{theorem}

Theorems \ref{T2.2}, \ref{T2.3}, and \ref{T2.4} are results on the different manifolds, however, the common result is that the particles are aligned on to a common geodesic.

\section{Velocity alignment models on the flat torus $\bbt^d$}\label{sec:3}
A velocity alignment model on general manifolds was suggested in \cite{ha2020emergent} and given as \eqref{A-2}. If the following a-priori condition is guaranteed in system \eqref{A-2}:
\begin{center}
``A shortest geodesic between two points $x_i(t)$ and $x_j(t)$ is unique for all $i\neq j$, $t\geq0$,"
\end{center}
then the solution of the system is well-defined and unique. Since the system contains the parallel transport on the manifold, the a-priori condition should be assumed for the well-posedness. However, if the domain $M$ has a cut locus(i.e. there exist at least two shortest geodesic between two points $x, y\in M$), then the system defined on $M$ is not well-defined in general. Since we choose the shortest geodesic, we guess that singularities come from the \textit{speciality of the shortest geodesic}. Regarding this reason, in this paper, we suggest a modified system which uses all geodesics to define the interactions. In this section, we suggest a modified velocity alignment model on the flat torus $\bbt^d$.

\subsection{Geodesics on the flat torus}\label{sec:3.1}
We are planning to consider all geodesics between $x_i,x_j\in M$ to define the interaction between $i^{th}$ and $j^{th}$ particles. So, in this subsection, we study the geodesics on the flat torus $\bbt^d$. We consider the flat torus $\mathbb{T}^d$ as a quotient space of the Euclidean space
\[
\mathbb{T}^d\simeq \bbr^d/_\sim,
\] 
where the equivalence relation $\sim$ given as 
\[
[x_1, x_2, \cdots, x_d]\sim[x_1+n_1, x_2+n_2,\cdots, x_d+n_d]\quad\forall~(n_1, \cdots, n_d)\in\bbz^d.
\]
Note that the metric of the torus was induced from $\bbr^d$.

\subsubsection{The case when $d=1$}
Before we study general dimension, we consider the simplest case. Since the universal covering space of $\bbt^1$ is $\bbr$, we can define a covering map $p$ as follows:
\[
p:\bbr\to\bbt^1,\quad x\mapsto x+\bbz.
\]
Let $x, y\in\bbt^1$ be two points on the one dimensional flat torus and $\gamma$ be a one of geodesic which connects two points $x$ and $y$. Then there exists two points $\tilde{x}, \tilde{y}\in\bbr$ such that $\gamma=p([\tilde{x}, \tilde{y}])$. Now we will express the following set:
\[
\Gamma_x^y:=\{\text{geodesics connecting $x$ and $y$}\}.
\]
From the previous argument, we can express the set of geodesics which connects $x$ and $y$ as follows:
\[
\tilde{\gamma}:[0, 1]\to\bbr^1,\quad \tilde{\gamma}(t)=\tilde{x}+t(\tilde{y}-\tilde{x}),\quad \gamma=p\circ \tilde{\gamma}.
\]
Since $\tilde{x}$ and $\tilde{y}$ are not unique, we can express all geodesics $\gamma_n$ with $n\in\bbz$ as follows:
\begin{align*}
&\tilde{\gamma}_n: [0, 1]\to \bbr^1,\quad \tilde{\gamma}_n(t)=\tilde{x}+t(\tilde{y}-\tilde{x}+n),\\
&\gamma_n: [0, 1]\to\bbt^1,\quad \gamma_n=p\circ \tilde{\gamma}_n.
\end{align*}
Then $\Gamma_{x}^y=\{\gamma_n\}_{n\in\bbz}$ is the set of geodesic from $x$ to $y$.

\begin{remark}\label{R3.1}
From the above definition, we have the following properties.
\begin{enumerate}
\item The length of $\gamma_n$ is given as
\[
|\gamma_n|=|\tilde{y}-\tilde{x}+n|.
\]
\item Since $\bbr^1/\bbt^1\simeq \bbz$, we can find a correspondence between $\Gamma_x^y$ and $\bbz$.
\end{enumerate}

\end{remark}
\subsubsection{The general case $d\geq2$}
Now we make some arguments on $\bbt^d$ case. Recall that the universal covering of $\bbt^d$  is $\bbr^d$ with the following covering map $p$:
\[
p:\bbr^d\to\bbt^d,\quad x\mapsto x+\bbz^d.
\]
From a similar argument that we made in the case of $d=1$, we can express $\Gamma_x^y$ as follows:
\[
\Gamma_x^y=\{\gamma_{\vec{n}}\}_{\vec{n}\in\bbz^d},
\]
where $\vec{n}\in\bbz^d$ and
\begin{align}
\begin{aligned}\label{C-0-1}
&\tilde{\gamma}_{\vec{n}}:[0, 1]\to\bbr^d,\quad \tilde{\gamma}_{\vec{n}}(t)=\tilde{x}+t(\tilde{y}-\tilde{x}+\vec{n}),\\
&\gamma_{\vec{n}}:[0,1]\to\bbt^d,\quad \gamma_{\vec{n}}=p\circ \tilde{\gamma}_{\vec{n}}.
\end{aligned}
\end{align}

\begin{center}
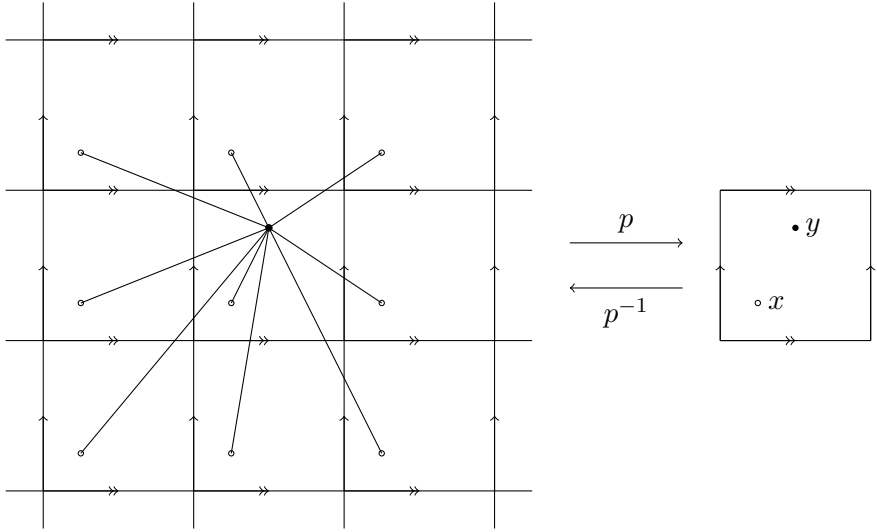
\begin{figure}[h]
\begin{tikzpicture}
\draw[-] (-2.5, 0)--(4.5, 0);
\draw[-] (-2.5, 2)--(4.5, 2);
\draw[-] (-2.5, 4)--(4.5, 4);
\draw[-] (-2.5, -2)--(4.5, -2);
\draw[-] (-2, -2.5)--(-2, 4.5);
\draw[-] (0, -2.5)--(0, 4.5);
\draw[-] (2, -2.5)--(2, 4.5);
\draw[-] (4, -2.5)--(4, 4.5);
\draw[->>] (0, 0)--(1, 0);
\draw[->>] (0, -2)--(1, -2);
\draw[->>] (0, 2)--(1, 2);
\draw[->>] (0, 4)--(1, 4);
\draw[->>] (-2, 0)--(-1, 0);
\draw[->>] (-2, -2)--(-1, -2);
\draw[->>] (-2, 2)--(-1, 2);
\draw[->>] (-2, 4)--(-1, 4);
\draw[->>] (2, 0)--(3, 0);
\draw[->>] (2, -2)--(3, -2);
\draw[->>] (2, 2)--(3, 2);
\draw[->>] (2, 4)--(3, 4);
\draw[->] (0, 0)--(0, 1);
\draw[->] (2, 0)--(2, 1);
\draw[->] (-2, 0)--(-2, 1);
\draw[->] (4, 0)--(4, 1);
\draw[->] (0, -2)--(0, -1);
\draw[->] (2, -2)--(2, -1);
\draw[->] (-2, -2)--(-2, -1);
\draw[->] (4, -2)--(4, -1);
\draw[->] (0, 2)--(0, 3);
\draw[->] (2, 2)--(2, 3);
\draw[->] (-2, 2)--(-2, 3);
\draw[->] (4, 2)--(4, 3);
\draw (0.5,0.5) circle (1pt);
\draw (0.5,0.5-2) circle (1pt);
\draw (0.5,0.5+2) circle (1pt);
\draw (0.5-2,0.5) circle (1pt);
\draw (0.5-2,0.5-2) circle (1pt);
\draw (0.5-2,0.5+2) circle (1pt);
\draw (0.5+2,0.5) circle (1pt);
\draw (0.5+2,0.5-2) circle (1pt);
\draw (0.5+2,0.5+2) circle (1pt);
\filldraw (1, 1.5) circle (1.2pt);
\draw[-] (1, 1.5)--(0.5,0.5);
\draw[-] (1, 1.5)--(0.5,0.5-2) ;
\draw[-] (1, 1.5)--(0.5,0.5+2);
\draw[-] (1, 1.5)--(0.5-2,0.5);
\draw[-] (1, 1.5)--(0.5-2,0.5-2);
\draw[-] (1, 1.5)--(0.5-2,0.5+2);
\draw[-] (1, 1.5)--(0.5+2,0.5);
\draw[-] (1, 1.5)--(0.5+2,0.5-2);
\draw[-] (1, 1.5)--(0.5+2,0.5+2);

\draw[-] (7, 0)--(9, 0);
\draw[->>] (7, 0)--(8, 0);
\draw[-] (7, 2)--(9, 2);
\draw[->>] (7, 2)--(8, 2);
\draw[-] (7, 0)--(7, 2);
\draw[->] (7, 0)--(7, 1);
\draw[-] (9, 0)--(9, 2);
\draw[->] (9, 0)--(9,1);
\draw (7.5,0.5) circle (1pt) node [right] {$x$};
\filldraw (8,1.5) circle (1pt) node [right] {$y$};

\draw[->] (5, 1.3)--(6.5, 1.3);
\draw[->] (6.5, 0.7)--(5, 0.7);
\draw (5.75, 1.3) node [above]{$p$};
\draw (5.75, 0.7) node [below]{$p^{-1}$};
\end{tikzpicture}
\caption{The universal covering of the two dimensional flat torus $\tilde{\mathbb{T}}^2$ on $\bbr^2$. White points are $p^{-1}(x)$ and a black point is an element of $p^{-1}(y)$. Drawn line segments are geodesic connecting the black point and white points.}
\label{Torus-Univ}
\end{figure}
\end{center}

\begin{remark}\label{R3.2}
From the above definition, we have the following properties.
\begin{enumerate}
\item The length of $\gamma_{\vec{n}}$ with $\vec{n}\in\bbz^d$ is given as
\[
|\gamma_{\vec{n}}|=\|\tilde{y}-\tilde{x}+\vec{n}\|.
\]
\item Since $\bbr^d/\bbt^d\simeq\bbz^d$, we can find a correspondence between $\Gamma_x^y$ and $\bbz^d$.
\end{enumerate}
\end{remark}
See Figure \ref{Torus-Univ} for the case of $d=2$, we express the geodesics of the flat torus on its universal covering space. We expressed all of the geodesics between $x,y\in M$ as follows:
\begin{align}\label{C-0-1-1}
\Gamma_x^y=\big\{
p([\tilde{x}_0,\tilde{y}]): \text{ for a fixed }\tilde{x}_0\in p^{-1}(x)\text{ and for all }\tilde{y}\in p^{-1}(y)
\big\},
\end{align}
where $[\tilde{x}_0, \tilde{y}]$ is a line segment on $\bbr^d$ which connects $\tilde{x}_0$ and $\tilde{y}$. Now, we are ready to construct a modified model on the flat torus.

\subsection{Construction of the modified model on the flat torus}\label{sec:3.2}
Before we define a system on the flat torus, we define a system on its covering space $\bbr^d$. Since we expressed the set of geodesics between two points $x, y$ as \eqref{C-0-1-1}, we can construct a system which considers all geodesics between two particles as follows: 
\begin{align}\label{C-0-2}
\begin{cases}
\displaystyle\frac{d}{dt}\tilde{x}_i=\tilde{v}_i,\vspace{0.2cm}\\
\displaystyle\frac{D}{dt}\tilde{v}_i=\frac{\kappa}{N}\sum_{k=1}^N\sum_{\tilde{y}_k\in p^{-1}(p(\tilde{x}_k))} \psi(\tilde{x}_i, \tilde{y}_k)(P_{\tilde{x}_i\tilde{y}_k}\tilde{v}_k-\tilde{v}_i),\quad t>0,\\
\tilde{x}_i(0)=\tilde{x}_i^0\in\bbr^d,\quad \tilde{v}_i(0)=\tilde{v}_i^0\in T_{\tilde{x}_i^0}\bbr^d,\quad\forall~i\in\mathcal{N},
\end{cases}
\end{align} 
where $P_{\tilde{x}_i\tilde{y}_k}$ is a parallel transport from a tangent vector at $\tilde{y}_k$ to a tangent vector at $\tilde{x}_k$. Here, we assume that the sum $\sum_{\tilde{y}_k\in p^{-1}(p(\tilde{x}_k))}$ is well-defined. We will discuss conditions for $\psi$ in the last of this subsection. Since the flat torus has zero curvatures, the parallel transport $P_{\tilde{x}_i\tilde{y}_k}$ is not necessary in this case. However, for a future generalization on general manifolds, we leave $P_{\tilde{x}_i\tilde{y}_k}$. Now, we assume that $\psi(x, y)$ is a function of $\mathrm{dist}_{\mathbb{R}^d}(x, y)$ and define a function $\varphi:\bbr_{\geq0}\to\bbr_{\geq0}$ as follows:
\[
\varphi(\mathrm{dist}_{\mathbb{R}^d}(x, y))=\psi(x, y)\quad\forall~x, y\in \bbt^d.
\]
If we put $x_i=p(\tilde{x}_i)$ and $v_i=Dp(\tilde{v}_i)$ for all $i\in\mathcal{N}$, we can reduce system \eqref{C-0-2} to $M$ as follows:
\begin{align}\label{C-0-3}
\begin{cases}
\displaystyle\frac{d}{dt}x_i=v_i,\vspace{0.2cm}\\
\displaystyle\frac{D}{dt}v_i=\frac{\kappa}{N}\sum_{k=1}^N \sum_{\gamma\in \Gamma_{x_k}^{x_i}}\varphi(|\gamma|)(P^\gamma_{ik} v_k-v_i),\quad t>0,\\
x_i(0)=x_i^0\in\bbt^d,\quad v_i(0)=v_i^0\in T_{x_i^0}\bbt^d,\quad \forall~i\in\mathcal{N},
\end{cases}
\end{align}
where $P^\gamma_{ik}$ is a parallel transport from a tangent vector at $x_k$ to a tangent vector at $x_i$ along a curve $\gamma$ and $|\gamma|$ is the length of $\gamma$. Recall that $P_{ik}^{\gamma_1}=P_{ik}^{\gamma_2}$ for any $\gamma_1, \gamma_2\in \Gamma_{x_k}^{x_i}$, since the domain of this system is the flat torus. So, we denote that $P_{ik}:=P_{ik}^\gamma$ for all $\gamma\in \Gamma_{x_k}^{x_i}$. If we define 
\begin{align}\label{C-0-3-1}
\Phi(x_i, x_k):=\sum_{\gamma\in \Gamma_{x_k}^{x_i}}\varphi(|\gamma|)
\end{align}
for all $i, k\in\mathcal{N}$, then system \eqref{C-0-3} can be expressed as follows:
\begin{align}\label{C-0-4}
\begin{cases}
\displaystyle\frac{d}{dt}x_i=v_i,\vspace{0.2cm}\\
\displaystyle\frac{D}{dt}v_i=\frac{\kappa}{N}\sum_{k=1}^N \Phi(x_i, x_k)(P_{ik} v_k-v_i),\quad t>0,\\
x_i(0)=x_i^0\in\bbt^d,\quad v_i(0)=v_i^0\in T_{x_i^0}\bbt^d,\quad \forall~i\in\mathcal{N}.
\end{cases}
\end{align}

Now, we discuss conditions for the communication function $\psi$ to guarantee the well-posedness of system \eqref{C-0-2}. Actually, it is equivalent to find a condition for the convergence of $\Phi$ in \eqref{C-0-3-1}. From the definition of $\Phi$ and Remark \ref{R3.2} (1), we have the following calculation:
\begin{align*}
\Phi(x, y)=\sum_{\gamma\in\Gamma_y^x}\varphi(|\gamma|)=\sum_{\vec{n}\in\bbz^d}\varphi(\| \tilde{y}-\tilde{x}+\vec{n}\|).
\end{align*}
Now we assume the following two conditions for $\varphi$:\\

\noindent($\mathcal{A}1$) A function $\varphi:\bbr_{\geq0}\to\bbr_{\geq0}$ is a continuous decreasing function.\vspace{0.2cm}

\noindent($\mathcal{A}2$) A sum $\displaystyle\sum_{\vec{n}\in\bbz^d}\varphi(\|\tilde{y}-\tilde{x}+\vec{n}\|)$ converges.\\

Here, we assumed $(\mathcal{A}1)$, since $\psi$ is a communication function. $(\mathcal{A}2)$ is assumed for the well-definedness of system \eqref{C-0-2} as we mentioned before. For a further argument, we provide the following lemma.
\begin{lemma}\label{L3.1}
Let $\tilde{x}, \tilde{y}\in\bbr^d$, then we have the following inequality:
\[
\min_{\vec{n}\in \bbz^d}\|\tilde{y}-\tilde{x}+\vec{n}\|\leq  \frac{\sqrt{d}}{2}.
\]
\end{lemma}

\begin{proof}
From a simple calculation, we get
\[
\min_{\vec{n}\in\bbz^d}\|\tilde{y}-\tilde{x}+\vec{n}\|^2=((\tilde{y})_1-(\tilde{x})_1+(\vec{n})_1)^2+\cdots+((\tilde{y})_d-(\tilde{x})_d+(\vec{n})_d)^2.
\]
Now, we pick $\vec{m}\in\bbz^d$ as follows:
\begin{align}\label{C-0-5}
(\vec{m})_i=\begin{cases}
[(\tilde{x})_i-(\tilde{y})_i]\qquad &\text{if }\quad0\leq((\tilde{x})_i-(\tilde{y})_i)-[(\tilde{x})_i-(\tilde{y})_i]<\frac{1}{2},\\
[(\tilde{x})_i-(\tilde{y})_i]+1\qquad &\text{if }\quad \frac{1}{2}\leq ((\tilde{x})_i-(\tilde{y})_i)-[(\tilde{x})_i-(\tilde{y})_i]<1
\end{cases}
\end{align}
where $[x]$ is the greatest integer which is not greater than $x$. Then we know that $-\frac{1}{2}\leq (\tilde{y})_i-(\tilde{x})_i+(\vec{m})_i\leq \frac{1}{2}$ for all $1\leq i\leq N$ and this yields
\[
\min_{\vec{n}\in\bbz^d}\|\tilde{y}-\tilde{x}+\vec{n}\|^2\leq\|\tilde{y}-\tilde{x}+\vec{m}\|^2\leq d\times \left(\frac{1}{2}\right)^2.
\]
Finally we can conclude that
\[
\min_{\vec{n}\in\bbz^d}\|\tilde{y}-\tilde{x}+\vec{n}\|\leq \frac{\sqrt{d}}{2}.
\]
This is the desired result.
\end{proof}

\subsection{Condition of the communication functions $\varphi$}\label{sec:3.3}
In this subsection, we study the condition for $\varphi$ which satisfies 
\[
\Phi(x, y)=\sum_{\vec{n}\in\bbz^d}\varphi(\|\tilde{y}-\tilde{x}+\vec{n}\|)<\infty\quad\forall~x, y\in M.
\]
Since $\varphi$ is continuous, non-negative and decreasing,  we can apply the integral test on $\bbr^d$ to get
\[
\sum_{\vec{n}\in\bbz^d}\varphi(\|\tilde{y}-\tilde{x}+\vec{n}\|)\quad\text{converges}\quad\Leftrightarrow\quad \int_{\bbr^d}\varphi(|x|)dx<\infty.
\]
From the radial symmetry, we have
\[
\int_{\bbr^d}\varphi(|x|)dx=|\bbs^{d-1}|\int_0^\infty r^{d-1}\varphi(r)dr,
\]
where $|\mathbb{S}^{d-1}|$ is a $(d-1)$-dimensional Hausdorff measure of the unit $(d-1)$-dimensional sphere. Actually, $|\mathbb{S}^{d-1}|$ can be expressed as
\begin{align}\label{C-10}
|\mathbb{S}^{d-1}|=\frac{d\pi^{d/2}}{\Gamma\left(\frac{d}{2}+1\right)},
\end{align}
where $\Gamma$ is the gamma function. Since \eqref{C-10} only depends on $d$, we have
\[
\sum_{\vec{n}\in\bbz^d}\varphi(\|\tilde{y}-\tilde{x}+\vec{n}\|)\quad\text{converges}\quad\Leftrightarrow\quad\int_0^\infty r^{d-1}\varphi(r)dr<\infty.
\]
So we can conclude that the condition ($\mathcal{A}1$) and ($\mathcal{A}2$) is equivalent to the following condition.\\

\noindent($\mathcal{A}$): A function $\varphi:\bbr_{\geq0}\to\bbr_{\geq0}$ is a continuous decreasing function and satisfies $\int_0^\infty r^{d-1}\varphi(r)dr$ exists.\\

Now, we study some examples of communication functions $\varphi$ which satisfy ($\mathcal{A}$).

\begin{example}
\noindent(1) When $\varphi$ has a compact support: 
Let $\mathrm{supp}\varphi\subset[0, A]$. Then of course, $\int_{\bbr^d}r^{d-1}\varphi(r)dr<\infty$.\\

\noindent(2) When $\varphi(r)=e^{-r}$: 
From a simple calculation, we have
\begin{align*}
\int_0^\infty r^{d-1}e^{-r}dr&=[-r^{d-1}e^{-r}]_0^\infty+(d-1)\int_0^\infty r^{d-2}e^{-r}dr\\
&=\cdots\\
&=(d-1)!\int_0^\infty e^{-r}dr=(d-1)!<\infty.
\end{align*}
If $d=1$, we can calculate the explicit form of $\Phi(x, y)$. From a similar argument that we used in \eqref{C-0-5}, we can choose $\tilde{x}\in p^{-1}(x)$ and $\tilde{y}\in p^{-1}(y)$ which satisfy $0\leq \delta:=\tilde{x}-\tilde{y}<1$. Then we have
\begin{align*}
\Psi(x, y)&=\sum_{n\in\bbz}e^{-|\tilde{x}-\tilde{y}+n|}=\sum_{n\in\bbz}e^{-|\delta+n|}=\sum_{n=0}^\infty e^{-(\delta+n)}+\sum_{n=1}^\infty e^{(\delta-n)}\\
&=e^{-\delta}\frac{1}{1-e^{-1}}+e^{\delta}\frac{e^{-1}}{1-e^{-1}}=\frac{e^{1-\delta}+e^{\delta}}{e-1}=\frac{\sinh(\delta-1/2)}{\sinh(1/2)}.
\end{align*}

\noindent(3) When $\varphi(r)=\frac{1}{(1+r^2)^{\alpha}}$ with $\alpha>\frac{d}{2}$: 
From a simple calculation, we have
\begin{align*}
\int_1^\infty \frac{r^{d-1}}{(1+r^2)^\alpha}dr<\int_1^\infty r^{d-1-2\alpha}dr<\infty.
\end{align*}
\end{example}

\section{Velocity alignment models on general manifold}\label{sec:4}
In this section, we generalize the result of Section \ref{sec:3} to a general manifold $M$. 

\subsection{Construction of the modified model on general manifold}\label{sec:4.1}
We can simply generalize system \eqref{C-0-3} on general manifold $M$ as follows:
\begin{align}\label{D-0}
\begin{cases}
\displaystyle\frac{d}{dt}x_i=v_i,\vspace{0.2cm}\\
\displaystyle\frac{D}{dt}v_i=\frac{\kappa}{N}\sum_{k=1}^N\sum_{\gamma\in\Gamma_{x_k}^{x_i}}\varphi(|\gamma|)(P_{ik}^\gamma v_k-v_i),\quad t>0,\\
x_i(0)=x_i^0\in M,\quad v_i(0)=v_i^0\in T_{x_i^0}M,\quad \forall i\in \mathcal{N}.
\end{cases}
\end{align}
We only have to concern about the well-definedness and the convergence of the sum $\sum_{\gamma\in \Gamma_{x_k}^{x_i}}$. Since arguing these issues on the universal covering space is easier, we lift system \eqref{D-0} defined on $M$ to a system on its universal covering space. We bring system \eqref{C-0-2} on general manifold as follows:
\begin{align}\label{D-1}
\begin{cases}
\displaystyle\frac{d}{dt}\tilde{x}_i=\tilde{v}_i,\vspace{0.2cm}\\
\displaystyle\frac{D}{dt}\tilde{v}_i=\frac{\kappa}{N}\sum_{k=1}^N\sum_{\tilde{y}_k\in p^{-1}(p(\tilde{x}_k))} \varphi(\mathrm{dist}(\tilde{x}_i, \tilde{y}_k))(P_{\tilde{x}_i\tilde{y}_k} \tilde{u}^{\tilde{y}_k}_k-\tilde{v}_i),\quad t>0,\\
\tilde{x}_i(0)=\tilde{x}_i^0\in\tilde{M}^d,\quad \tilde{v}_i(0)=\tilde{v}_i^0\in T_{\tilde{x}_i^0}\tilde{M},\quad\forall~i\in\mathcal{N},
\end{cases}
\end{align} 
where $p^{-1}_{\tilde{x}}$ is a inverse map of $p$ defined on neighborhood of $\tilde{x}$ and $\tilde{u}_k^{\tilde{y}_k}=Dp^{-1}_{\tilde{y}_k}(v_k)=Dp^{-1}_{\tilde{y}_k}\circ Dp_{\tilde{x}_k}(\tilde{v}_k)$ and $P_{\tilde{x}\tilde{y}}$ is a parallel transport from a tangent vector at $\tilde{y}$ to a tangent vector at $\tilde{x}$. Since a sum $\sum_{\tilde{y}_k\in p^{-1}(p(\tilde{x}_k))}$ should be well-defined, we assume the follows:\\

\noindent($\mathcal{M}1$): A set $p^{-1}(x)\subset \tilde{M}$ is at most countable set for any $x\in M$.\\

When $M=\bbt^d$, we have already showed that there exists a one-to-one correspondence between $p^{-1}(x)$ and $\bbz^d$ in Section \ref{sec:3.1}. So, in this case, obviously $p^{-1}(x)$ is countable for any $x\in \bbt^d$. On the other hand, we know that $P_{\tilde{x}_i\tilde{y}_k}$ only depends on two points $\tilde{x}_i$ and $\tilde{y}_k$. We should assume the following property of $\tilde{M}$:\\

\noindent($\mathcal{M}2$): For any points $\tilde{x}, \tilde{y}\in\tilde{M}$, there exists a unique geodesic which connects two points $\tilde{x}$ and $\tilde{y}$.\\

If $\tilde{M}$ is Euclidean space or has constant negative curvature(Hyperbolic space), then it satisfies ($\mathcal{M}2$). System \eqref{D-1} given on a manifold $M$ is well defined for any $M$ which satisfies ($\mathcal{M}1$) and ($\mathcal{M}2$). Especially, if $\tilde{M}=\bbr^d$, then the parallel transport $P_{\tilde{x}_i\tilde{y}_k}$ in \eqref{D-1} can be omitted. In this case, analyzing system \eqref{D-1} is easier than analyzing system \eqref{D-0}, since we do not have to consider the parallel transport on $M$.
In Section \ref{sec:5}, we study some examples of $M$(e.g. flat torus, flat M\"{o}bius strip, and flat Klein bottle) which of the universal covering space is the Euclidean spaces.

\subsection{Reduction to the original CS model on manifold}\label{sec:4.2}
In this section, we introduce the relationship between system \eqref{D-0} and previous velocity alignment models. If we consider that system \eqref{D-0} is defined on the Euclidean space(i.e. $M=\bbr^d$), then there is the unique geodesic $\gamma$ which connects two points $x, y\in M$ and its length is $\mathrm{dist}(x, y)$. So the set of geodesics $\Gamma_x^y$ only contains one element. This implies that system \eqref{D-0} can be reduced to the original CS system \eqref{A-1}.

Now, we compare two systems \eqref{A-2} and \eqref{D-0} given on $\bbt^1$. We define $\varphi$ which satisfies $\varphi\left(\frac{1}{2}\right)=0$. If $x, y\in \bbt^1$, then the length of geodesics are $|\tilde{y}-\tilde{x}+n|$ for any $n\in\bbz$. We know that there are at most one $m\in\bbz$ such that $|\tilde{y}-\tilde{x}+m|<\frac{1}{2}$. This implies that only the shortest geodesic between two points $x_i$ and $x_k$ determine the interaction between $i^{th}$ and $k^{th}$ particles. i.e. system \eqref{D-0} can be reduced as follows:
\begin{align*}
\begin{cases}
\displaystyle\frac{d}{dt}x_i=v_i,\vspace{0.2cm}\\
\displaystyle\frac{D}{dt}v_i=\frac{\kappa}{N}\sum_{k=1}^N\varphi(\mathrm{dist}(x_i, x_k))(P_{ik} v_k-v_i),\quad t>0,\\
x_i(0)=x_i^0\in \bbt^1,\quad v_i(0)=v_i^0\in T_{x_i^0}\bbt^1,\quad \forall i\in \mathcal{N},
\end{cases}
\end{align*}
where $P_{ik}$ is a parallel transport from the tangent space at $x_i$ to the tangent space at $x_k$ along shortest geodesic between them. This system is exactly same with system \eqref{A-2}. So, we can conclude that a modified velocity alignment system \eqref{D-0} can be reduced to a previous system \eqref{A-2}.

\subsection{Emergent behaviors of the modified model}\label{sec:4.3}
In this subsection, we study the emergent behaviors of system \eqref{D-0}. We define the following energy functional:
\begin{align*}
\mathcal{E}(\mathcal{V})=\frac{1}{2}\sum_{i=1}^N\|v_i\|^2,
\end{align*}
where $\mathcal{V}=\{v_i\}_{i=1}^N$.
\begin{lemma}\label{L4.1}
Let $(\mathcal{X}, \mathcal{V})$ be a solution of system \eqref{D-0}. Then we have
\begin{align*}
\frac{d}{dt}\mathcal{E}(\mathcal{V})=-\frac{\kappa}{2N}\sum_{i, k=1}^N\sum_{\gamma\in \Gamma_{x_k}^{x_i}}\psi(|\gamma|)\|P_{ik}^\gamma v_k-v_i\|^2.
\end{align*}
i.e. the energy of the system is a non-increasing function.
\end{lemma}

\begin{proof}
From the simple calculation, we have
\begin{align*}
\frac{d}{dt}\mathcal{E}(\mathcal{V})&=\sum_{i=1}^N\left\langle \frac{D}{dt}v_i, v_i\right\rangle\\
&=\frac{\kappa}{N}\sum_{i, k=1}^N\sum_{\gamma\in \Gamma_{x_k}^{x_i}}\psi(|\gamma|)\langle P_{ik}^\gamma v_k-v_i, v_i\rangle\\
&=\frac{\kappa}{2N}\sum_{i, k=1}^N\sum_{\gamma\in \Gamma_{x_k}^{x_i}}\psi(|\gamma|)\langle P_{ik}^\gamma v_k-v_i, v_i\rangle+\frac{\kappa}{2N}\sum_{i, k=1}^N\sum_{\gamma\in \Gamma_{x_i}^{x_k}}\psi(|\gamma|)\langle P_{ki}^\gamma v_i-v_k, v_k\rangle\\
&=\frac{\kappa}{2N}\sum_{i, k=1}^N\sum_{\gamma\in \Gamma_{x_k}^{x_i}}\psi(|\gamma|)\langle P_{ik}^\gamma v_k-v_i, v_i\rangle+\frac{\kappa}{2N}\sum_{i, k=1}^N\sum_{\gamma\in \Gamma_{x_k}^{x_i}}\psi(|\gamma|)\langle P_{ki}^{-\gamma} v_i-v_k, v_k\rangle\\
&=\frac{\kappa}{2N}\sum_{i, k=1}^N\sum_{\gamma\in \Gamma_{x_k}^{x_i}}\psi(|\gamma|)\Big(
\langle P_{ik}^\gamma v_k-v_i, v_i\rangle+\langle P_{ki}^{-\gamma} v_i-v_k, v_k\rangle
\Big)\\
&=-\frac{\kappa}{2N}\sum_{i, k=1}^N\sum_{\gamma\in \Gamma_{x_k}^{x_i}}\psi(|\gamma|)\|P_{ik}^\gamma v_k-v_i\|^2,
\end{align*}
where $-\gamma$ stands for the inverse path of $\gamma$. We have the desired result.
\end{proof}

\begin{remark}
It is easy to show that the second temporal derivative of $\mathcal{E}(\mathcal{V})$ is bounded.
\end{remark}

We combine Lemma \ref{L4.1} and Barbalat's lemma(Lemma \label{L2.1}) to obtain the following theorem.
\begin{theorem}\label{T4.1}
Let $(\mathcal{X}, \mathcal{V})$ be a solution of system \eqref{D-0}. Then we have
\begin{align*}
\lim_{t\to\infty}\sum_{\gamma\in \Gamma_{x_k}^{x_i}}\psi(|\gamma|)\|P_{ik}^\gamma v_k-v_i\|^2=0
\end{align*}
for any $i, k\in\mathcal{N}$.
\end{theorem}

From this theorem, we have the following corollary.
\begin{corollary}\label{C4.1}
Let $(\tilde{\mathcal{X}}, \tilde{\mathcal{V}})$ be a solution of system \eqref{D-1}. Then we have
\begin{align*}
\lim_{t\to\infty}\sum_{\tilde{y}_k\in p^{-1}(p(\tilde{x}_k))}\varphi(\mathrm{dist}(\tilde{x}_i, \tilde{y}_k))\|P_{\tilde{x}_i\tilde{y}_k}\tilde{u}_k^{\tilde{y}_k}-\tilde{v}_i\|^2=0,
\end{align*}
where $P_{\tilde{x}_i\tilde{y}_k}\tilde{u}_k^{\tilde{y}_k}$ is defined in \eqref{D-1}.
\end{corollary}

Since given models \eqref{D-0} and \eqref{D-1} are defined on too general manifold $M$, we can not provide more specific emergent behaviors. In Section \ref{sec:5}, we provide more specific emergent behaviors of system \eqref{D-0} and \eqref{D-1} on the specific manifold $M$.

\subsection{Self-interaction effect}
In a previous velocity alignment model on manifold \eqref{A-2}, $i^{th}$ particle does not interact with itself, since $P_{ii}$ is the identity map and this yields $P_{ii}v_i-v_i=0$. However, system \eqref{D-0} is different. Since we considered all geodesics which is connecting $x_i$ and $x_j$ to define the interaction between $i^{th}$ and $j^{th}$ particles, if there exists non-trivial geodesic starts from $x_i$ and finish at $x_i$ then $i^{th}$ particle interacts with itself. In this subsection, we focus on this effect. We put $N=1$ and $(x, v):=(x_1, v_1)$ to system \eqref{D-0}. Then we get
\begin{align}\label{D-10}
\begin{cases}
\displaystyle\frac{d}{dt}x=v,\vspace{0.2cm}\\
\displaystyle\frac{D}{dt}v=\kappa\sum_{\gamma\in \Gamma_x^x}\varphi(|\gamma|)(P^\gamma v-v),\quad t>0,\\
x(0)=x^0\in M,\quad v(0)=v^0\in T_{x^0}M,
\end{cases}
\end{align}
where $P^\gamma$ is a parallel transport along $\gamma$. From Lemma \ref{L4.1}, we know that
\begin{align}\label{D-11}
\frac{d}{dt}\left(\frac{1}{2}\|v\|^2\right)=-\frac{\kappa}{2}\sum_{\gamma\in \Gamma_x^x}\varphi(|\gamma|)\|P^\gamma v-v\|^2.
\end{align}
This yields
\[
\frac{d}{dt}\|v\|^2=-\kappa\sum_{\gamma\in \Gamma_x^x}\varphi(|\gamma|)\|P^\gamma v-v\|^2\geq -4\kappa\sum_{\gamma\in \Gamma_x^x}\varphi(|\gamma|)\|v\|^2,
\] 
and from a simple fact of ODE, we know that if $v^0\neq0$, then $v(t)\neq0$ for all $t\geq0$. So we can define $u(t)$ for all $t\geq0$ as follows:
\begin{align}\label{D-12}
u(t):=\frac{v(t)}{\|v(t)\|}.
\end{align}
Now we substitute \eqref{D-12} into \eqref{D-11} to get
\[
\frac{d}{dt}\|v\|^2=-\kappa\|v\|^2\sum_{\gamma\in \Gamma_x^x}\varphi(|\gamma|)\|P^\gamma u-u\|^2
\]
or equivalently,
\[
\frac{d}{dt}\ln\|v\|=-\frac{\kappa}{2}\sum_{\gamma\in \Gamma_x^x}\varphi(|\gamma|)\|P^\gamma u-u\|^2.
\]
From this, we can conclude that the self-interaction effect reduces the speed of particle, and the ratio is at most exponential. If we apply Theorem \ref{T4.1} to system \eqref{D-10}, then we get
\begin{align}\label{D-13}
\lim_{t\to\infty}\|v\|^2\sum_{\gamma\in \Gamma_x^x}\varphi(|\gamma|)\|P^\gamma u-u\|^2=0.
\end{align}
Finally, we can obtain a dichotomy for the long-time behaviors of system \eqref{D-10}:

\noindent(1) The speed converges to zero. i.e.
\[
\lim_{t\to\infty}\|v\|=0.
\]
(2) The direction of the velocity $u$ satisfies
\[
\lim_{t\to\infty}\sum_{\gamma\in \Gamma_x^x}\varphi(|\gamma|)\|P^\gamma u-u\|^2=0.
\]
This effect comes from the topology of the domain. This argument has not been argued before.

\section{Systems on quotient spaces of the Euclidean space}\label{sec:5}
In this section, we study system \eqref{D-0} on $M$ which is quotient spaces of the Euclidean space. We also assume that the universal covering space of $M$ is the Euclidean space $\bbr^d$. If we consider system \eqref{D-0} on its universal covering space $\tilde{M}=\bbr^d$, then we have system \eqref{D-1}. Since $\tilde{M}=\bbr^d$, we can omit the parallel transport $P_{\tilde{x}_i\tilde{y}_k}$ in system \eqref{D-1}. So we have
\begin{align}\label{E-1}
\begin{cases}
\displaystyle\frac{d}{dt}\tilde{x}_i=\tilde{v}_i,\vspace{0.2cm}\\
\displaystyle\frac{d}{dt}\tilde{v}_i=\frac{\kappa}{N}\sum_{k=1}^N\sum_{\tilde{y}_k\in p^{-1}(p(\tilde{x}_k))} \varphi(\mathrm{dist}(\tilde{x}_i, \tilde{y}_k))( Dp^{-1}_{\tilde{y}_k}\circ Dp_{\tilde{x}_k}(\tilde{v}_k)-\tilde{v}_i),\quad t>0,\vspace{0.2cm}\\
\tilde{x}_i(0)=\tilde{x}_i^0\in\tilde{M}^d,\quad \tilde{v}_i(0)=\tilde{v}_i^0\in T_{\tilde{x}_i^0}\tilde{M},\quad\forall~i\in\mathcal{N}.
\end{cases}
\end{align}
Here we use the original temporal derivative of $\tilde{v}_i$, since the domain is the Euclidean space. The only complicate thing in system \eqref{E-1} is calculating $Dp_{\tilde{y}}^{-1}\circ Dp_{\tilde{x}}$ for $\tilde{x}, \tilde{y}\in p^{-1}(x)$. In this section, we provide some specific spaces which are quotient spaces of the Euclidean space(e.g. The flat torus, M\"{o}bius strip, and the flat Klein bottle), and calculate the explicit form of $Dp_{\tilde{y}}^{-1}\circ Dp_{\tilde{x}}$ to obtain more specific emergent behaviors of the modified velocity alignment system.

\subsection{Example 1: The flat torus $\bbt^d$}\label{sec:5.1}
We consider the flat torus $\bbt^d$ in this subsection. We have already studied the covering map $p$ between $\bbr^d$ and $\bbt^d$ in Section \ref{sec:3.1} as follows:
\[
p(\tilde x)=((\tilde x)_1-[(\tilde x)_1], (\tilde x)_2-[(\tilde x)_2],\cdots, (\tilde x)_d-[(\tilde x)_d])\quad\forall~\tilde x\in\bbr^d.
\]
This yields that $Dp^{-1}_{\tilde{y}}\circ Dp_{\tilde{x}}$ is the identity map. Also, we have already express the set of geodesics between $x,y\in\bbt^d$ as \eqref{C-0-1}. From the above results, we can reduce system \eqref{E-1} as follows:
\begin{align}\label{E-2}
\begin{cases}
\displaystyle\frac{d}{dt}\tilde{x}_i=\tilde{v}_i,\vspace{0.2cm}\\
\displaystyle\frac{d}{dt}\tilde{v}_i=\frac{\kappa}{N}\sum_{k=1}^N\left(\sum_{\vec{n}\in\bbz^d} \varphi(\mathrm{dist}(\tilde{x}_i, \tilde{x}_k+\vec{n}))\right)( \tilde{v}_k-\tilde{v}_i),\quad t>0,\vspace{0.2cm}\\
\tilde{x}_i(0)=\tilde{x}_i^0\in\bbr^d,\quad \tilde{v}_i(0)=\tilde{v}_i^0\in \bbr^d,\quad\forall~i\in\mathcal{N}.
\end{cases}
\end{align}

Now, we apply Corollary \ref{C4.1}, we can obtain the following theorem.
\begin{proposition}\label{P5.1}
Let $(\mathcal{X}, \mathcal{V})$ be a solution of system \eqref{D-0} defined on the flat torus $\bbt^d$. Let $\varphi$ satisfies $(\mathcal{A})$ and 
\[
\varphi\left(\frac{\sqrt{d}}{2}\right)>0.
\]
Then, we have
\begin{align*}
\lim_{t\to\infty}(\tilde{v}_k-\tilde{v}_i)=0\quad\forall~i, k\in\mathcal{N}.
\end{align*}
\end{proposition}

\begin{proof}
We apply Corollary \ref{C4.1} to system \eqref{E-2} to get
\begin{align}\label{E-3}
\lim_{t\to\infty}\sum_{\vec{n}\in\bbz^d} \varphi(\mathrm{dist}(\tilde{x}_i, \tilde{x}_k+\vec{n})\|\tilde{v}_k-\tilde{v}_i\|^2=0.
\end{align}
Since $\varphi\geq0$, we also have
\begin{align*}
\sum_{\vec{n}\in\bbz^d} \varphi(\mathrm{dist}(\tilde{x}_i, \tilde{x}_k+\vec{n})\geq\inf_{\vec{n}\in\bbz^d} \varphi(\mathrm{dist}(\tilde{x}_i, \tilde{x}_k+\vec{n})).
\end{align*}
Now we apply Lemma \ref{L3.1} to get
\[
\sum_{\vec{n}\in\bbz^d} \varphi(\mathrm{dist}(\tilde{x}_i, \tilde{x}_k+\vec{n})\geq \varphi\left(\frac{\sqrt{d}}{2}\right).
\]
Now, we substitute the above relation into \eqref{E-3} to get
\[
 \varphi\left(\frac{\sqrt{d}}{2}\right)\lim_{t\to\infty}\|\tilde{v}_k-\tilde{v}_i\|^2=0.
\]
Since $ \varphi\left(\frac{\sqrt{d}}{2}\right)>0$, we get
\[
\lim_{t\to\infty}\|\tilde{v}_k-\tilde{v}_i\|=0.
\]
\end{proof}

\begin{remark}[Self-interaction effect of the flat torus]
Since $P^\gamma$ is identity map for any closed curve $\gamma$, we have $P^\gamma v-v=0$ for any tangent vector $v$. So, relation \eqref{D-13} implies nothing in this case. So we can conclude that there is no self-interaction effect on the flat torus. 
\end{remark}

\subsection{Example 2: The flat M\"{o}bius strip $\mathbb{M}$}\label{sec:5.2}
We consider the flat M\"{o}bius strip in this subsection and denote it as $\mathbb{M}$. We consider that $\mathbb{M}$ is given on $[0, 1]\times\bbr\subset\bbr^2$ with the identification $(0, t)\sim(1, -t)$ where $t\in\bbr$. We can express the universal covering $\tilde{\mathbb{M}}=\bbr^2$ with the following covering map $p: \tilde{\bbm}\to[0, 1)\times\bbr\simeq\bbm$:
\begin{align*}
p(\tilde{x})=\begin{cases}
((\tilde{x})_1-[(\tilde{x})_1], (\tilde{x})_2),\quad &\text{if $[(\tilde{x})_1]$ is even},\\
((\tilde{x})_1-[(\tilde{x})_1],-(\tilde{x})_2),\quad&\text{if $[(\tilde{x})_1]$ is odd},
\end{cases}\quad\forall~ \tilde{x}\in \bbr^2,
\end{align*}
where $[x]$ is the largest integer not greater than $x\in\bbr$. From the above covering map $p$, we get
\[
Dp^{-1}_{\tilde{y}}\circ Dp_{\tilde{x}} (\tilde{v})=\begin{cases}
((\tilde{v})_1, (\tilde{v})_2),\quad &\text{if }[(\tilde{y})_1]-[(\tilde{x})_1]\text{ is even},\\
((\tilde{v})_1, -(\tilde{v})_2),\quad &\text{if }[(\tilde{y})_1]-[(\tilde{x})_1]\text{ is odd}.
\end{cases}
\]
For the notation simplicity, we define the operator $J$ as follows:
\[
J(\tilde{x})=((\tilde{x})_1, -(\tilde{x})_2)\quad \forall~ \tilde{x}\in\bbr^2.
\]
We also denote $J^n$ be the $n$-th power of the operator $J$. Then, we can express $p^{-1}(p(\tilde{x}))$ as follows:
\[
p^{-1}(p(\tilde{x}))=\{J^n(\tilde{x})+(n, 0): n\in \bbz\}
\]
and
\[
Dp^{-1}_{\tilde{y}}\circ Dp_{\tilde{x}} (\tilde{v})=J^n(\tilde{v})
\]
if $\tilde{y}=J^n(\tilde{x})+(n, 0)$ for some $n\in\bbz$. Using these notations, we can reduce system \eqref{E-1} into
\begin{align}\label{E-5}
\begin{cases}
\displaystyle\frac{d}{dt}\tilde{x}_i=\tilde{v}_i,\vspace{0.2cm}\\
\displaystyle\frac{d}{dt}\tilde{v}_i=\frac{\kappa}{N}\sum_{k=1}^N\sum_{n\in\bbz} \varphi(\mathrm{dist}(\tilde{x}_i, J^n(\tilde{x}_k)+(n, 0))(J^n(\tilde{v}_k)-\tilde{v}_i),\quad t>0,\\
\tilde{x}_i(0)=\tilde{x}_i^0\in\bbr^2,\quad \tilde{v}_i(0)=\tilde{v}_i^0\in \bbr^2,\quad\forall~i\in\mathcal{N}.
\end{cases}
\end{align}
From a similar argument that we made in Section \ref{sec:3.3}, we can prove
\begin{align}\label{E-5-1}
\sum_{n\in\bbz}\varphi\big(\mathrm{dist}(\tilde{x}, J^n(\tilde{y})+(n, 0))\big)<\infty\quad\forall~\tilde{x},\tilde{y}\in\bbr^2\quad\Longleftrightarrow\quad \int_0^\infty \varphi(r)dr<\infty
\end{align}
from the integral test. If we apply Corollary \ref{C4.1} to system \eqref{E-5}, we get
\begin{align}\label{E-6}
\lim_{t\to\infty}\sum_{n\in\bbz}\varphi(\mathrm{dist}(\tilde{x}_i, J^n(\tilde{x}_k)+(n, 0))\|J^n(\tilde{v}_k)-\tilde{v}_i\|=0.
\end{align}
Since $J^2=\mathrm{Id}$, we have $J^{2m+1}=J$ and $J^{2m}=\mathrm{Id}$ for all integer $m$. From this fact, we get
\begin{align*}
&\sum_{n\in\bbz}\varphi(\mathrm{dist}(\tilde{x}_i, J^n(\tilde{x}_k)+(n, 0))\|J^n(\tilde{v}_k)-\tilde{v}_i\|^2\\
&=\left(\sum_{m\in\bbz}\varphi(\mathrm{dist}(\tilde{x}_i, \tilde{x}_k+(2m, 0))\right)\|\tilde{v}_k-\tilde{v}_i\|^2\\
&\hspace{5cm}+\left(\sum_{m\in\bbz}\varphi(\mathrm{dist}(\tilde{x}_i, J(\tilde{x}_k)+(2m+1, 0))\right)\|J(\tilde{v}_k)-\tilde{v}_i\|^2.
\end{align*}
We combine the above result and \eqref{E-6} to obtain
\begin{align}\label{E-7}
\begin{cases}
\displaystyle\lim_{t\to\infty}\left(\sum_{m\in\bbz}\varphi(\mathrm{dist}(\tilde{x}_i, \tilde{x}_k+(2m, 0))\right)\|\tilde{v}_k-\tilde{v}_i\|=0,\\
\displaystyle\lim_{t\to\infty}\left(\sum_{m\in\bbz}\varphi(\mathrm{dist}(\tilde{x}_i, J(\tilde{x}_k)+(2m+1, 0))\right)\|J(\tilde{v}_k)-\tilde{v}_i\|=0,
\end{cases}
\end{align}
for all $i, k\in \mathcal{N}$. Now we have the following lemma.

\begin{lemma}\label{L5.1}
Let $p$ be a covering map between two dimensional the Euclidean space and the flat M\"{o}bius strip. 
Let $x, y\in [0, 1]\times \bbr/_\sim\simeq \mathbb{M}$ and $|x_2|, |y_2|<L$ for some positive number $L$. Then, we have
\begin{align*}
\inf_{m\in\bbz}\mathrm{dist}(\tilde{x},\tilde{y}+(2m, 0))<\sqrt{1+4L^2},\quad
\inf_{m\in\bbz}\mathrm{dist}(\tilde{x},J(\tilde{y})+(2m+1, 0))<\sqrt{1+4L^2},
\end{align*}
where $\tilde{x}\in p^{-1}(x), \tilde{y}\in p^{-1}(y)$. 
\end{lemma}

\begin{proof}
First, we fix two points $x, y\in\mathbb{M}$ and we choose $\tilde{x},\tilde{y}\in\bbr^2$ such that $\tilde{x}\in p^{-1}(x)$ and $\tilde{y}\in p^{-1}(y)$:
\[
\tilde{x}=(\tilde{x}_1, \tilde{x}_2),\quad \tilde{y}=(\tilde{y}_1, \tilde{y}_2).
\]
Without loss of generality, we can assume that $-1\leq \tilde{x}_1-\tilde{y}_1<1$. Then we have
\[
\mathrm{dist}(\tilde{x}, \tilde{y}+(2m, 0))^2=(\tilde{x}_1-\tilde{y}_1-2m)^2+(\tilde{x}_2-\tilde{y}_2)^2\leq(\tilde{x}_1-\tilde{y}_1)^2+(\tilde{x}_2-\tilde{y}_2)^2.
\]
Here, the last equality only holds for $m=0$. Since $-1\leq\tilde{x}_1-\tilde{y}_1<1$ and $-2L<\tilde{x}_2-\tilde{y}_2<2L$, we have
\[
\inf_{m\in\bbz}\mathrm{dist}(\tilde{x}, \tilde{y}+(2m, 0))^2<1+4L^2.
\]
Similarly, we have
\begin{align*}
\inf_{m\in\bbz}\mathrm{dist}(\tilde{x}, J(\tilde{y})+(2m+1, 0))^2&=\inf_{m\in\bbz}\left((\tilde{x}_1-\tilde{y}_1-2m-1)^2+(\tilde{x}_2+\tilde{y}_2)^2\right)<1+4L^2.
\end{align*}
\end{proof}

From this lemma and \eqref{E-7}, we get the following proposition.

\begin{proposition}\label{P5.2}
Let $(\mathcal{X}, \mathcal{V})$ be a solution of system \eqref{D-0} defined on $M=\mathbb{M}$. If we assume the following a-priori assumption: ``the second component of $x_i\in[0, 1)\times \bbr$ lies in $(-L, L)$ for all $t\geq0$ and $i\in\mathcal{N}$." We also assume that $\varphi$ satisfies \eqref{E-5-1} and $\varphi(\sqrt{1+4L^2})>0$. Then, we have
\[
\lim_{t\to\infty}(\tilde{v}_i-\tilde{v}_k)=0\quad\forall~i, k\in\mathcal{N},
\]
and the second components of all velocities converges to zero. i.e.
\[
\lim_{t\to\infty}(\tilde{v}_i)_2=0\quad\forall i\in\mathcal{N}.
\]
\end{proposition}

\begin{proof}

From Lemma \ref{L5.1}, we have
\begin{align}
\begin{aligned}\label{E-8}
\sum_{m\in\bbz}\varphi(\mathrm{dist}(\tilde{x}_i, \tilde{x}_k+(2m, 0))&\geq \inf_{m\in\bbz}\varphi(\mathrm{dist}(\tilde{x}_i, \tilde{x}_k+(2m, 0))\geq \varphi(\sqrt{1+4L^2}),\\
\sum_{m\in\bbz}\varphi(\mathrm{dist}(\tilde{x}_i, J(\tilde{x}_k)+(2m+1, 0))&\geq \inf_{m\in\bbz}\varphi(\mathrm{dist}(\tilde{x}_i, J(\tilde{x}_k)+(2m+1, 0))\geq \varphi(\sqrt{1+4L^2}).
\end{aligned}
\end{align}
Here, we used the given a-priori assumption. We substitute \eqref{E-8} into \eqref{E-7} to obtain
\[
\varphi(\sqrt{1+4L^2})\lim_{t\to\infty}\|\tilde{v}_k-\tilde{v}_i\|=0,\quad \varphi(\sqrt{1+4L^2})\lim_{t\to\infty}\|J(\tilde{v}_k)-\tilde{v}_i\|=0.
\]
Since $\varphi(\sqrt{1+4L^2})>0$, we have 
\begin{align}\label{E-9}
\lim_{t\to\infty}\|\tilde{v}_k-\tilde{v}_i\|=0,\quad \lim_{t\to\infty}\|J(\tilde{v}_k)-\tilde{v}_i\|=0.
\end{align}
We consider the case when $k=i$ of second equality for \eqref{E-9} to get
\[
\lim_{t\to\infty}\|J(\tilde{v}_i)-\tilde{v}_i\|=0\quad\forall i\in\mathcal{N}.
\]
From the definition of the operator $J$, we could obtain that the second component of $\tilde{v}_i$ converges to zero.
\end{proof}

\begin{remark}[Self-interaction effect of the flat M\"{o}bius strip]\label{R5.2}
If we put $x_i=x_k=x$ into \eqref{E-6}, we get \eqref{D-13} on the M\"{o}bius strip as follows:
\[
\lim_{t\to\infty}\sum_{n\in\bbz}\varphi(\mathrm{dist}(\tilde{x}, J^n(\tilde{x})+(n, 0)))\|J^n(\tilde{v})-\tilde{v}\|=0.
\]
From the definition of $J$, we can simplify the above relation as follows:
\[
\lim_{t\to\infty}\left(\sum_{m\in\bbz}\varphi(\mathrm{dist}(\tilde{x}, J(\tilde{x})+(2m+1, 0)))\right)\|J(\tilde{v})-\tilde{v}\|=0.
\]
The simple calculation yields
\[
1\leq\inf_{m\in\bbz}\mathrm{dist}(\tilde{x}, J(\tilde{x})+(2m+1, 0))=\inf_{m\in\bbz}\sqrt{(2m+1)^2+4(\tilde{x})_2^2}\leq \sqrt{1+4L^2}.
\]
If $\varphi(1)=0$, then $\sum_{m\in\bbz}\varphi(\mathrm{dist}(\tilde{x}, J(\tilde{x})+(2m+1, 0)))=0$ and there is no self-interaction effect. However, if we assume $\varphi(\sqrt{1+4L^2})>0$ and a-priori condition: $|x^2|<L$, then we have
\[
\lim_{t\to\infty}\|J(\tilde{v})-\tilde{v}\|=0.
\]
This also implies
\[
\lim_{t\to\infty}(\tilde{v})_2=0.
\]
So, in this case, the self-interaction effect of the flat M\"{o}bius strip yields that the second coordinate of the velocity converges to zero.
\end{remark}

\subsection{Example 3: The flat Klein bottle $\mathbb{K}$}\label{sec:5.3}
We consider the flat Klein bottle in this subsection and denote it as $\bbk$. We consider that the $\bbk$ is given on $[0, 1]^2\subseteq\bbr^2$ as Figure \ref{Fig3}.
\begin{center}
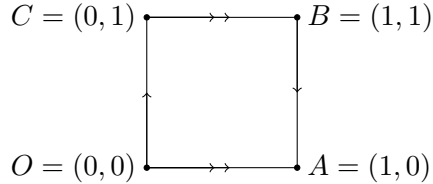
\begin{figure}[h]
\begin{tikzpicture}
\coordinate (O) at (0,0);
\draw[-] (O) -- (2, 0) -- (2, 2)--(0, 2)--(O);
\draw[->] (O)-- (0, 1);
\draw[->] (2, 2)-- (2, 1);
\draw[->] (0, 0)-- (0.9, 0);
\draw[->] (0, 0)-- (1.1, 0);
\draw[->] (0,2)-- (0.9, 2);
\draw[->] (0, 2)-- (1.1, 2);
\filldraw (O) circle (1pt) node[left] {$O=(0, 0)$};
\filldraw (2, 0) circle (1pt) node[right] {$A=(1, 0)$};
\filldraw (2, 2) circle (1pt) node[right] {$B=(1, 1)$};
\filldraw (0, 2) circle (1pt) node[left] {$C=(0, 1)$};
\end{tikzpicture}
\caption{The flat Klein bottle $\mathbb{K}$ on $\bbr^2$}
\label{Fig3}
\end{figure}
\end{center}
Here, two directed line segments $OC$ and $BA$ are identified as $(0, t)\sim(1, 1-t)$ and other two directed line segments $OA$ and $CB$ are identified as $(t, 0)\sim(t, 1)$ where $0\leq t\leq 1$. We can express the universal covering $\tilde{\bbk}$ of $\bbk$ on $\bbr^2$. Also, we can express the covering map $p:\tilde{\bbk}=\bbr^2\to[0, 1)^2\subset\bbk$ as follows:
\begin{align*}
p(\tilde{x})=\begin{cases}
((\tilde{x})_1-[(\tilde{x})_1], (\tilde{x})_2-[(\tilde{x})_2]),\quad &\text{if $[(\tilde{x})_1]$ is even},\\
((\tilde{x})_1-[(\tilde{x})_1], -(\tilde{x})_2-[-(\tilde{x})_2]),\quad &\text{if $[(\tilde{x})_1]$ is odd},
\end{cases}
\quad \tilde{x}\in\bbr^2
\end{align*}

From this, we get
\[
Dp^{-1}_{\tilde{y}}\circ Dp_{\tilde{x}} (\tilde{v})=\begin{cases}
((\tilde{v})_1, (\tilde{v})_2),\quad &\text{if }[(\tilde{y})_1]-[(\tilde{x})_1]\text{ is even},\\
((\tilde{v})_1, -(\tilde{v})_2),\quad &\text{if }[(\tilde{y})_1]-[(\tilde{x})_1]\text{ is odd}.
\end{cases}
\]
Now, we use $J$ defined in Section \ref{sec:5.2} to get
\[
p^{-1}(p(\tilde{x}))=\{J^n\tilde{x}+(n, m):(n, m)\in\bbz^2\}.
\]
We also have that if $\tilde{y}=J^n(\tilde{x})+(n, m)$ for some $(n, m)\in\bbz^2$, then
\[
Dp_{\tilde{y}}^{-1}\circ Dp_{\tilde{x}}(\tilde{v})=J^n(\tilde{v}).
\]
This notations allow us to simplify system \eqref{E-1} defined on $\mathbb{M}$ as follows:
\begin{align}\label{E-10}
\begin{cases}
\displaystyle\frac{d}{dt}\tilde{x}_i=\tilde{v}_i,\vspace{0.2cm}\\
\displaystyle\frac{D}{dt}\tilde{v}_i=\frac{\kappa}{N}\sum_{k=1}^N\sum_{(n, m)\in\bbz^2} \varphi\big(\mathrm{dist}(\tilde{x}_i, J^n(\tilde{x}_k)+(n, m))\big)(J^n(\tilde{v}_k)-\tilde{v}_i),\quad t>0,\\
\tilde{x}_i(0)=\tilde{x}_i^0\in\bbr^2,\quad \tilde{v}_i(0)=\tilde{v}_i^0\in \bbr^2,\quad\forall~i\in\mathcal{N}.
\end{cases}
\end{align} 
Again $J^2=\mathrm{Id}$, we have
\begin{align*}
&\sum_{(n, m)\in\bbz^2}\varphi(\mathrm{dist}(\tilde{x}_i, J^n(\tilde{x}_k)+(n, m)))\|J^n(\tilde{v}_k)-\tilde{v}_i\|^2\\
&=\sum_{(l, m)\in\bbz^2}\varphi(\mathrm{dist}(\tilde{x}_i, \tilde{x}_k+(2l, m)))\|\tilde{v}_k-\tilde{v}_i\|^2
+\sum_{(l, m)\in\bbz^2}\varphi(\mathrm{dist}(\tilde{x}_i, J(\tilde{x}_k)+(2l+1, m)))\|J(\tilde{v}_k)-\tilde{v}_i\|^2
\end{align*}

Now, we apply Corollary \ref{C4.1} to get 
\begin{align}\label{E-11}
\begin{cases}
\displaystyle\lim_{t\to\infty}\sum_{(l, m)\in\bbz^2}\varphi(\mathrm{dist}(\tilde{x}_i, \tilde{x}_k+(2l, m)))\|\tilde{v}_k-\tilde{v}_i\|^2=0,\\
\displaystyle\lim_{t\to\infty}\sum_{(l, m)\in\bbz^2}\varphi(\mathrm{dist}(\tilde{x}_i, J(\tilde{x}_k)+(2l+1, m)))\|J(\tilde{v}_k)-\tilde{v}_i\|^2=0,
\end{cases}
\end{align}
for all $i, k\in \mathcal{N}$. 
Now, we have the following lemma.
\begin{lemma}\label{L5.2}
Let $\tilde{x}, \tilde{y}\in\bbr^2$. Then we have
\begin{align*}
\inf_{(l, m)\in\bbz^2}\mathrm{dist}(\tilde{x}, \tilde{y}+(2l, m))\leq\frac{\sqrt{5}}{2},\quad
\inf_{(l, m)\in\bbz^2}\mathrm{dist}(\tilde{x}, J(\tilde{y})+(2l+1, m))\leq\frac{\sqrt{5}}{2}.
\end{align*}

\end{lemma}
Since the proof of the above lemma is similar to the proof of Lemmas \ref{L3.1} and Lemma \ref{L5.1}, we omitted it. From this lemma, we have the following proposition.
\begin{proposition}
Let $(\mathcal{X}, \mathcal{V})$ be a solution of system \eqref{D-0} defined on $M=\bbk$. If $\varphi$ satisfies $(\mathcal{A})$ and $\varphi\left(\frac{\sqrt{5}}{2}\right)>0$, then we have
\[
\lim_{t\to\infty}(\tilde{v}_i-\tilde{v}_k)=0\quad\forall~i, k\in\mathcal{N}
\]
and the second component of $\tilde{v}_i$ converges to zero. i.e.
\[
\lim_{t\to\infty}(\tilde{v}_i)_2=0\quad\forall~i\in\mathcal{N}.
\]
\end{proposition}

\begin{proof}
From a similar argument that we used in the proof of Proposition \ref{P5.2}, we can prove this proposition with \eqref{E-7} and Lemma \ref{L5.2}
\end{proof}

\begin{remark}[Self-interaction effect of the flat Klein bottle]
If we put $x_i=x_k=x$ into \eqref{E-11}, we get \eqref{D-13} on the flat Klein bottle as follows:
\[
\lim_{t\to\infty}\sum_{(l, m)\in\bbz^2}\varphi(\mathrm{dist}(\tilde{x}, J(\tilde{x})+(2l+1, m)))\|J(\tilde{v})-\tilde{v}\|^2=0
\]
The simple calculation yields
\[
1\leq\inf_{(l, m)\in\bbz^2}\mathrm{dist}(\tilde{x}, J(\tilde{x})+(2l+1, m))=\inf_{(l, m)\in\bbz^2}\sqrt{(2l+1)^2+(2\tilde{x}^2-m)^2}\leq \frac{\sqrt{5}}{2}.
\]
From a similar argument that we made in Remark \ref{R5.2}, if $\varphi(1)=0$, then there is no self-interaction effect. However, if $\varphi\left(\frac{\sqrt{5}}{2}\right)>0$, then the flat Klein bottle yields the self-interaction effect which makes that the second coordinate of the velocity converges to zero.
\end{remark}

\section{Conclusion}\label{sec:6}
Throughout this paper, we constructed a velocity alignment model on manifolds and studied the emergent dynamics in various spaces. Previously in \cite{ha2020emergent}, the CS model on manifold uses the shortest geodesic which connects two particles to define an interaction between them. One disadvantage of this model is that if there exist two or more than two shortest geodesics then the system can not be well-defined. We tried to remove the speciality of the shortest geodesic. So, we consider all of geodesics which is connecting $i^{th}$ particle and $j^{th}$ particle to define an interaction between $i^{th}$ and $j^{th}$ particles. In this way, we could obtain a modified velocity alignment system. One of the interesting properties of a modified system is the self-interaction effect. This effect depends on the structure of the domain of the system. In the last section, we provided the emergent behaviors of the modified system defined on some specific spaces(flat torus, flat M\"{o}bius strip, flat Klein bottle). However, our model can not be applied on the sphere $\bbs^d$ with $d\geq2$. Since there are uncountably many geodesics that connecting antipodal pair points, we can not define the sum of interactions in this case. We need another setup to cover this case, and it can be nice future work.

\bibliographystyle{abbrv}
\def\url#1{}
\bibliography{lit}

\end{document}